\documentclass[aps,pre,twocolumn,groupedaddress]{revtex4-1}

\usepackage{graphicx}
\usepackage{amsmath}
\usepackage{amssymb}
\usepackage{amsthm}
\usepackage{arydshln} 
\usepackage{tikz}
\usetikzlibrary{shapes}
\usetikzlibrary{arrows}
\usepackage{color}

\begin{document}

\newtheorem{corollary}{Corollary}[section]
\newtheorem{lemma}{Lemma}[section]
\newtheorem{theorem}{Theorem}[section]
\newtheorem{definition}{Definition}[section]


\title{Comparison of Communities Detection Algorithms for Multiplex}


\author{Chuan Wen, Loe$^1$ and Henrik Jeldtoft Jensen$^2$}
\affiliation{}


\date{\today}

\begin{abstract}
Multiplex is a set of graphs on the same vertex set, i.e. $\{G(V,E_1),\ldots,G(V,E_m)\}$. It is a generalized graph to model multiple relationships with parallel edges between vertices. This paper is a literature review of existing communities detection algorithms for multiplex and a comparative analysis of them.
\end{abstract}


\maketitle

\section{Introduction}
Complexity Science studies the collective behaviour of a system of interacting agents, and a graph (network) is often an apt representation to visualize such systems. Traditionally the agents are expressed as the vertices, and an edge between a vertex pair implies that there are interactions between them. 

However the modern outlook in Network Science is to generalize the edges to encapsulate the multiple type of relationships between agents. For example in a social network, people are acquainted through work, school, family, etc. This is to preserve the richness of the data and to reveal deeper perspectives of the system. This is known as a \emph{multiplex}.

Multiplex is a natural transition from graphs and many disciplines independently studied this mathematical model for various applications like communities detection. A community refers to a set of vertices that behaves differently from the rest of the system. This is to modularize a complex system into simpler representations to form the bigger picture of the information flow.

The first half of this paper is a literature review of the different communities detection algorithms and some theoretical bounds on graph cutting. Next we propose a suite of benchmark multiplexes and similarity metrics to determine the similarity of various communities detection algorithms. Finally we present the empirical results for this paper.

\section{Preliminaries}
\begin{definition}\label{def:multiplex}\textbf{(Multiplex)}
A multiplex is a finite set of $m$ graphs, $\mathcal{G} = \{G^1,\ldots,G^m\}$, where every graph $G^i = (V,E_i)$ has a distinct edge set $E_i \subseteq V \times V$.
\end{definition}

There are many synonymous names for multiplex and occasionally they simultaneously used in the same paper. This is to assist readers to visualize the system through the different descriptions. For instance a \emph{``multi-layer network"} describes the multiplex as  layers of graphs, i.e. layer $i$ implies $G^i \in \mathcal{G}$. The following list is the synonymous names for multiplex and will be avoided for the rest of the paper: 
\textbf{Multigraph} \cite{Cormode05,gjoka11}, 
\textbf{MultiDimensional Network} \cite{Berlingerio11_2,Berlingerio11,Hossmann12,kazienko2011multidimensional,lerch2006cliques,RossettiBG11,tang2009}, 
\textbf{Multi-Relational Network} \cite{cai_et_al_2005,DavisLC11,Rodriguez201029,Szell03082010,BryKWF12,LouatiHP12,Skvoretz07}, 
\textbf{MultiLayer Network} \cite{brodka12_2,brodka2011degree,brodka2013introduction,Li12,Cozzo12,mucha2010community}, 
\textbf{PolySocial Networks} \cite{Fischer}, 
\textbf{Multi-Modal Network} \cite{ambrosino2012hub,leblanc1988transit,Nagurney2000393}, 
\textbf{Heterogeneous  Networks} \cite{dong2012link} and 
\textbf{Multiple Networks} \cite{magnani2013formation}. 
However we will use the $i^{th}$ \emph{layer}, \emph{relationship} or \emph{dimension} to refer to the graph $G^i \in \mathcal{G}$. This is to help us to express certain ideas in a more concrete manner.

For example when layers of graphs are stacked on top of each other, there will be vertex pairs with edges ``overlapping" each other (Def. \ref{def:overlap}). The distribution to the number of the overlapping edges is an important characteristic of a multiplex, where it is used to classify different multiplex ensembles \cite{cellai13,Lee12,chen2012degree,RanolaASSL10,Bianconi13}.
\begin{definition}\textbf{(Overlapping Edges)} \label{def:overlap}
Let two edges from two different graphs in $\mathcal{G}$ be $e(u,v) \in E_i$ and $e'(u',v') \in E_j$, where $i \neq j$. The edges $e$ and $e'$ overlap if and only if $e=e'$, i.e. $u=u'$ and $v=v'$.
\end{definition}

If there arise ambiguity to the context, we will distinguish the ``community" between a multiplex and a monoplex (a graph in a multiplex) as \textbf{multiplex-community} and \textbf{monoplex-community} respectively. This will avoid confusion when we review the different multiplex communities detection algorithms.

Many of these multiplex-algorithms divide the multiplex problem into independent communities detection problems on the monoplexes. The solutions for these monoplexes, known as \textbf{auxiliary-partitions}, provide the supplementary information for the multiplex-algorithm to aggregate. The principal solution from the aggregation forms the \textbf{multiplex-partition}, which defines the communities in the multiplex.


\section{Definitions of a Multiplex-Community}\label{sec:quality_community}
A \emph{community} is vaguely described as a set of interacting agents that collectively behaves differently from its neighboring agents. However there is no universally accepted formal definition, since the construct of a community depends on the problem domain. \cite{fortunato2010community} categorized this diversity by the communities' \emph{Local Definitions}, \emph{Global Definitions} and \emph{Vertex Similarity}. 

\subsection{Local Definition}
From the assumption that a community has weak interactions with their neighboring vertices, the evaluation of a community can be isolated from the rest of the network. Thus it is sufficient to establish a community from the perspective of the members in the community.

Consider each graphs in a multiplex as an independent mode of communication between the members, e.g. email, telephone, postal, etc. A high quality community should resume high information flow amongst its members when one of the communication modes fails (1 less graph). Hence Berlingerio et al. proposed the \emph{redundancy} of the communities \cite{Berlingerio11_2} as a measure to the quality of a multiplex-community.

\begin{definition} \textbf{(Redundancy)}  Let $W \subseteq V$ be the set of vertices in a multiplex-community and $P \subseteq W \times W$ be the set of vertex pairs in $W$ that are adjacent in $\geq1$ relationship. The set of redundant vertex pairs are $P' \subseteq P$ where vertex pairs in $W$ that are adjacent in $\geq 2$ relationships. The redundancy of $W$ is determined by:
\begin{equation}\label{def:redundancy}
\frac{1}{|\mathcal{G}| \times |P|} \sum_{G^i \in \mathcal{G}} \sum_{\{u,v\} \in P'} \delta(u,v,E_i) ,
\end{equation}
where $\delta(u,v,E_i) = 1$ (zero otherwise) if $\{u,v\} \in E_i$.
\end{definition}

Eq. \ref{def:redundancy} counts the number of edges in the multiplex-community where their corresponding vertex pairs are adjacent in two or more graphs. The sum is normalized by the theoretical maximum number of edges between all adjacent vertex pairs, i.e. $({|\mathcal{G}| \times |P|})$. The quality of a multiplex-community is determined by how identical the subgraphs (induced by the vertices of the multiplex-community) are across the graphs in the multiplex. 

Thus the \emph{redundancy} does not depends on the number of edges in the multiplex-community, i.e. not a necessary condition to its quality. This can lead to an unusual idea that a community can be low in density. For instance a cycle of overlapping edges form a ``community" of equal quality as a complete clique of overlapping edges.

\subsection{Global Definition}\label{global_def}
The global measure of a partition considers the quality of the communities \emph{and} their interactions among the communities. For example the \emph{modularity} function by Newman and Girvan measures how different a monoplex-communities are from a random graph (Def. \ref{modularity}) \cite{newman2006modularity}.
\begin{definition}\label{modularity}
\textbf{(Modularity)} Let $A_{ij}$ be the adjacency matrix of a graph with $|E|$ edges and $k_i$ is the degree of vertex $i$. $\delta(v_i, v_j) = 1$ if $v_i$ and $v_j$ are in the same community, otherwise  $\delta(v_i, v_j) = 0$. The \emph{Modularity function} measures how far the communities differs from a random graph:
\begin{equation}\label{eq:modularity}
Q = \frac{1}{2|E|} \sum_{ij} \Big( A_{ij} - \frac{k_ik_j}{2|E|} \Big) \delta(v_i, v_j).
\end{equation}
\end{definition}

Given a fixed partition on the vertex set, the modularity on each of the $m$ graphs in the multiplex differs. Therefore a good multiplex-communities suggests that all the monoplex-communites in the graphs have high modularity.

To quantify this concept, Tang et. al claims that if there exists latent communities in the multiplex, a subset of the graphs in the multiplex, $\mathcal{G}' \subset \mathcal{G}$ has sufficient information to find these communities \cite{tang2009}. If the hypothesis is true, then the communities detected from $\mathcal{G}'$ should reflect high modularity on the rest of the graphs in the multiplex, i.e. $\mathcal{G} \setminus \mathcal{G}'$. 

In the language of machine learning, pick a random graph $G \in \mathcal{G}$ as the test data and let $\mathcal{G}' = \mathcal{G} \setminus G$ be the training data. The multiplex-partition $P$ yielded from a communities detection algorithm on $\mathcal{G}'$ is evaluated with the modularity function on the test data $G$. $P$ is a good multiplex-partition if the modularity of partition $P$ on the graph $G$ is maximized. This extends the modularity metric for multiplex.

\subsection{Vertex Similarity}
Two vertices belongs to a Vertex Similarity community if they are similar by some measure. For example the Edge Clustering Coefficient \cite{Radicchi} of a vertex pair in a graph measures the (normalized) number of common neighbors between them. A high Edge Clustering Coefficient implies that there are many common neighbors between the vertex pair, thus suggesting that the two vertices should belong in the same community. In the extension for multiplex, Brodka et. al introduced Cross-Layer Edge Clustering Coefficient (CLECC) \cite{brodka12_2}.

\begin{definition}\label{def:CLECC} \textbf{(Cross-Layer Edge Clustering Coefficient)} Given a parameter $\alpha$, the MIN-Neighbors of vertex $v$, $N(v,\alpha)$ are the set of vertices that are adjacent to $v$ in at least $\alpha$ graphs. The Cross-Layer Edge Clustering Coefficient of two vertices $u,v \in V$ measures the ratio of their common neighbors to all their neighbors.
\begin{equation}\label{eq:CLECC}
CLECC(u,v,\alpha) = \frac{|N(u,\alpha) \cap N(v,\alpha)|}{|N(u,\alpha) \cup N(v,\alpha) 
\setminus \{u,v\}|}.
\end{equation}
\end{definition}

A pair of vertices in a multiplex of social networks with low CLECC suggests that the individuals do not share a common clique of friends through at least $\alpha$ social networks. Therefore it is unlikely that they form a community together. 

\subsection{Densely Connected Community Cores}
Communities detection is the process to partition the vertices such that every vertex belongs in at least a community. However there are cases where one just desires some substructures in a multiplex with certain properties. For example a \emph{Dense Connected Community Core} is the set of vertices such that they are in the same community for all the auxiliary-partitions \cite{yin2013}.

\section{Theoretical Bounds}
The Max-Cut problem finds a partition of a graph such that the number of edges induced across the clusters are maximized. Thus the same partition over the complement graph minimizes the number of edges between the clusters. This is known as the Balanced-Min-Cut problem and it is closely related to our communities detection problem, where the number of edges induced between the communities are minimized. 

Therefore to extend our understanding for multiplex, we begin with a known result for the Max-Cut problem. It allows us to prove a corollary for the Balanced-Min-Cut problem on multiplex.

\subsection{Maximum Cut Problem on Multiplex}
\begin{theorem}\label{maxkcut}
Consider graph $G^1, \dots, G^m$ on the same vertex set $V$. There exists a k-partition of $V$ into $k \geq 2$ communities $C_1,\dots,C_k$ such that for all $i=1,\dots,m$:
\begin{equation}\label{maxcut}
\mbox{\# edges cut in $G^i$} \geq \frac{(k-1)|E_i|}{k} - \sqrt{2m|E_i|/k}.
\end{equation}
\end{theorem}
\begin{proof}
Sketch \cite{KuhnO07}: The proof for $k \geq 3$ is similar to the case for $k=2$, hence we will only prove the latter. WLOG, the vertex set of $G^1$ is partitioned into 2 equal sets, $A$ and $B$. Define $X_i$ be an indicator function where $X_i=1$ if the $i^{th}$ edge is induced from $A$ to $B$, $X_i=0$ if otherwise. Since $Pr(X_i)=1/2$ and with linearity of expectation, we get $\mathbb{E}[X_i]=|E_i|/2$ and $\mathbb{E}[X_i^2]=|E_i||E_i+1|/4$. We can now by use of Chebyshev's Inequality get the probability that a given partition fail the inequality \ref{maxcut} for $G^1$. Since the probability sum of all graphs to fail is less than 1, thus there exists a partition such that all the graphs fulfill inequality \ref{maxcut}.
\end{proof}

Further developments on the bounds were made by imposing additional conditions where the maximum degree is bounded \cite{KuhnO07} and in cases where $k=m=2$ \cite{Patel08}. Since the solution for Max-Cut on graphs is NP-complete, the extension to simultaneously Max-Cut all the graphs in a multiplex is naturally NP-complete too. Thus this also implies that balanced minimum bisection is NP-complete too \cite{Garey1990}. 

\subsection{Balanced Minimum Cut Problem on Multiplex}
\begin{corollary}\label{balminkcut}
Consider graph $G^1, \dots, G^m$ on the same vertex set $V$. There exists a k-partition of $V$ into $k \geq 2$ equal-sized communities $C_1,\dots,C_k$ (i.e. $|C_i|\approx|C_j|$) such that for all $i=1,\dots,m$:
\begin{equation}\label{minkcut}
\mbox{\# edges cut in $G^i$} \leq \frac{(k-1)|\bar{E}|}{k} - \sqrt{2m|\bar{E}|/k},
\end{equation}
where $|\bar{E}| = {n \choose 2} - |E_i|.$
\end{corollary}
\begin{proof}
Let $\bar{G}^i$ be the complement graph of $G^i$, and its edge set is denoted by $\bar{E}_i$. Since the maximum number of edges in a graph is ${n \choose 2}$, hence $|\bar{E}_i| = {n \choose 2} - |E_i|$. Apply (Max-Cut) Theorem \ref{maxkcut} on the set of complement graphs $\bar{G}_i$ and substitutes $|\bar{E}|_i$ into the result, the expression in the corollary follows. The proof in Theorem \ref{maxkcut} ensures that the communities are equal in size.
\end{proof}

A partition that fulfills Eq. \ref{minkcut} is not necessary a good community defined in Section \ref{sec:quality_community}, vice versa. However the edges induced between partition classes are often perceived as bottlenecks when information flows through the network/multiplex. They are similar to the bridges between cities and communities. Therefore Communities Detection Algorithms tend to minimize the number of edges between different communities.

\section{Communities Detection Algorithms for Multiplex}\label{sec:CDAlgo}
The general strategy for existing communities detection algorithm for multiplex is to extract features from the multiplex and reduce the problem to a familiar representation. In solving the reduced representation, the multiplex-communities are then deduced from the auxiliary solutions of the reduced problems.

Therefore many multiplex algorithms rely on existing monoplex-communities detection algorithms to get the auxiliary-partitions for the interim steps. The choice of algorithms is often independent of their extension for multiplex, and hence any communities detection algorithm in theory can be chosen to solve the interim steps. In this paper our experiments used Louvain Algorithm to generate the auxiliary-partitions.

\subsection{Projection}\label{sec:projection}
The naive method is to projected the multiplex into a weighted graph. I.e. let $A^i$ be the adjacency matrix of $G^i \in \mathcal{G}$. The adjacency matrix of the weighted projection of $\mathcal{G}$ is given by $\bar{A} = \frac{1}{m}\sum_{i=1}^m A^i$. We will call this the ``Projection-Average" of multiplex.

It was been independently proposed as a baseline for more sophisticated multiplex algorithms as the performance is often ``sub-par" \cite{Berlingerio11_2,kazienko2011multidimensional,Nefedov11,DavisLC11,RossettiBG11}. In our experiments we will compare this with the unweighted variant, that is the ``Projection-Binary" of a multiplex, i.e. $G(V,E_1 \cup \ldots \cup E_m)$.

An alternative weight assignment between vertex pair is to consider the connectivity of their neighbors, where a high ratio of common neighbors implies stronger ties \cite{Berlingerio11_2}. This is based on the idea that members of the same community tend to interact over the same subset of relations, which was independently proposed by Brodka et. al in Def. \ref{def:CLECC} \cite{brodka12_2}. This alternative will be known as ``Projection-Neighbors".

\subsection{Consensus Clustering}\label{sec:CCPAlgo}
The previous strategy aggregates the graphs first, and then it performs the communities detection algorithm over the resultant graph. It is a poor strategy as it neglects the rich information of the dimensions \cite{tang2009}. Therefore the Consensus Clustering strategy is to first apply the communities detection algorithm on the graphs separately as auxiliary partitions, and then the principal clustering (multiplex communities) is derived by aggregating these auxiliary partitions in a meaningful manner.

The key concept behind consensus clustering is to measure the frequency with which two vertices are found in the same community among the auxiliary partitions. Vertices that are frequently in the same monoplex-community are more likely to be in the same multiplex-community. Therefore the communities detection algorithm on the individual graphs on the multiplex determines the structural properties of multiplex-communities, whereas Consensus Clustering determines the relational properties of the multiplex-communities.

\subsubsection{Frequent Closed Itemsets Mining}
Data-mining is to find a set of items that occurs frequently together in a series of transactions. For example items like milk, cereal and fruits are frequently bought together in supermarkets based on a series of sales transactions. These sets are known as \emph{itemsets}. Berlingerio et. al translates the Consensus Clustering of the auxiliary-partitions as a data-mining problem to discover multiplex-communities \cite{berlingerio2013abacus}.

The vertices in the multiplex defines the $|V|$ \emph{transactions} for the data-mining, and the items are tuples $(c,d)$ where the respective vertex belongs in monoplex-community $c$ in dimension $d$. For example suppose vertex $v_i$ belongs to monoplex-communities $c_1,c_5$ and $c_2$ in dimensions $d_1,d_2$ and $d_3$ respectively. The $i^{th}$ transaction is the set of items $\{(c_1,d_1),(c_5,d_2),(c_2,d_3)\}$. Therefore when we use data-mining methods like \emph{Frequent Closed Itemsets Mining}, we are able to identify the frequent (relative to a predefined threshold) itemsets as multiplex-communities.

For example each vertex is a customer's transaction in a supermarket, and a community is a target market that the supermarket wants to discover. It is only meaningful if the target market is sufficiently large, and thus we need to defined the minimum community size (e.g. 10). In this case a customer's transaction is his auxiliary-communities membership. Therefore Frequent Closed Itemsets Mining will extract a multiplex-community on at least e.g. 10 vertices with which each customer's transaction is a subset of the target market's itemsets.

\subsubsection{Cluster-based Similarity Partitioning Algorithm}
 \emph{Cluster-based Similarity Partitioning Algorithm} averages the number of instances vertex pairs are in the same auxiliary-communities. For example in a multiplex with 5 dimensions, if there are 3 instances where vertices $v_i$ and $v_j$ are in the same auxiliary-community, then the similarity value of vertex pair $(v_i,v_j)$ is $3/5$.

Once the similarity is measured for all the vertex pairs, the principal cluster is determined with k-means clustering --- vertices with the closest similarity at each iteration are grouped together. Therefore vertex pairs that are frequently in the same auxiliary-communities will have high similarity value, and hence more likely to be clustered together in the same principal-community. This is known as Partition Integration by Tang et. al \cite{tang2009}.

\subsubsection{Generalized Canonical Correlations}
Each of the auxiliary-partitions maps the vertices as points in a $l$-dimensional (this dimension is independent to the dimensions of a multiplex) Euclidean space. The points are positioned in a way that the shorter the shortest path between two vertices are, the closer they are in the Euclidean space. One of such mapping can be achieved by  concatenating the top eigenvectors of the adjacency matrix. Thus given $d$ graphs in a multiplex, there are $d$ structural feature matrices $S^i$ of size $l \times n$ where the column in each matrix is the position of a vertex in the $l$-dimensional Euclidean space.

Tang et. al wants to aggregate the structural feature matrices to a principal structural feature matrix $\bar{S}$ such that the principal partition can be determined from $\bar{S}$ \cite{tang2009}. The ``average"  $\bar{S} = \frac{1}{d}\sum^{d}_{i=1}S^{(i)}$ however does not result in sensible principal structural feature matrix since the matrix elements between $S^{(i)}$ and $S^{(j)}$ are independent.

A solution to fix this problem is to transform the $S^{(i)}$ such that they are in the same space and their ``average" is sensible. That is the same vertex in the $d$ different Euclidean spaces are aligned in the same point in a Euclidean space. Specifically we need a set of linear transformations $w_i$ such that they maximize the pairwise correlations of the $S^{(i)}$, and Generalized Canonical Correlations Analysis is one of such standard statistical tools \cite{Kettenring}. This allows us to ``average" the structures in a more sensible way:
\begin{equation}
\bar{S} = \frac{1}{d}\sum^{d}_{i=1}S^{(i)}w_i.
\end{equation}

Finally the principal partition is determined via k-mean clustering of the principal feature matrix $\bar{S}$.

\subsection{Bridge Detection}\label{algo:bridge}
A bridge in a graph refers to an edge with high information flow, like the busy roads between two cities, where the absence of these roads separates the cities into isolated communities. One way to do this is to project the multiplex $\mathcal{G}$ to a weighted network and determine the bridges from the projection. Alternatively one can remove them by the definition of a multiplex-bridge to get the desired partitions. 

In social networks, strong edge ties are desirable within the communities. Hence to identify weak ties between vertex pairs, Brodka et. al proposed CLECC (Eq. \ref{eq:CLECC}) as a measure. At each iteration, all connected vertex pairs are recomputed and the pair with the lowest CLECC score will be disconnected in all the graphs. The algorithm halts when the desired number of communities (components) are yielded greedily \cite{brodka12_2}. This is the same strategy presented by Girvan and Newman, where the bridges of a graph were identified by their betweenness centrality score \cite{girvan2002community}. In the experiments, we set $\alpha = |\mathcal{G}|/2$ for the CLECC score in Eq. \ref{eq:CLECC}.

\subsection{Tensor Decomposition}\label{sec:tensor}
Algebraic Graph Theory is a branch of Graph Theory where algebraic methods like linear algebra are used to solve problems on graph. Hence the natural representation for a multiplex is a $3^{rd}$-order tensor (as a multidimensional array) instead of a matrix ($2^{nd}$-order tensor). The set of $m$ graphs in a multiplex is a set of $m$ $n \times n$ adjacency matrices, which can be represented as a $m \times n \times n$ multidimensional array (tensor) \cite{Rodriguez201029}. This allows us to leverage on the available tensor arithmetics like tensor decomposition.

Tensor decompositions are analogues to the singular value decomposition and 'Lower Upper' decomposition in matrices, where they express the tensor into simpler components. For example a PARAFAC tensor decomposition \cite{harshman1970fpp} is the rank-k approximation of a tensor $\mathcal{T}$ as a sum of rank-one tensors (vectors $\bar{u}^{(i)}$, $\bar{v}^{(i)}$ and $\bar{w}^{(i)}$), i.e.:
\begin{equation}\label{Eq:PARAFAC}
\mathcal{T} \approx \sum^k_{i=1} \bar{u}^{(i)} \circ \bar{v}^{(i)} \circ \bar{w}^{(i)}.
\end{equation}
where $\bar{a} \circ \bar{b}$ denotes the vector outer product. The components in the $i^{th}$ factor,  $\bar{u}^{(i)}$, $\bar{v}^{(i)}$ and $\bar{w}^{(i)}$, suggest that there are strong ties (possibly a cluster/community) between the elements in $\bar{u}^{(i)}$ and $\bar{v}^{(i)}$ via the dimension in the top component in $\bar{w}^{(i)}$.

For instance suppose the $j^{th}$ element in $\bar{w}^{(i)}$ is the most largest element. This suggests that in the $i^{th}$ community, the top 10 (or any predefined threshold) elements in $\bar{u}^{(i)}$ are in the same cluster as the top 10 elements in $\bar{v}^{(i)}$ via the $j^{th}$ dimension/relationship \cite{dunlavy,Judith,Andri,papalexakis}.

\section{Benchmark Multiplex}\label{sec:benchmark}
A \emph{Erd\H{o}s-R\'{e}nyi Graph} is a graph where vertex pairs are connected with a fixed probability \cite{Erdos60}. The random nature of this construction usually does not have any meaningful communities structures in them. Hence it is not ideal to use as a benchmark graph for Communities Detection Algorithms. 

Therefore a benchmark graph should be similar to the Girvan and Newman Model where some random edges are induced between a set of dense subgraphs (as communities) to form a single connected component/graph. The set of dense subgraphs acts as ``ground-truth" communities of the graph for a Communities Detection Algorithm to discover and hence a benchmark for Communities Detection Algorithms. The goal of this section is to design similar benchmarks for multiplex.

The main challenge is that there is not yet a universally accepted definition of a good multiplex community. Hence there is no methodology for us to construct a benchmark such that it does not favor certain algorithms. Therefore the objective of the following benchmark graphs is \textbf{to study the relations between these multiplex-communities detection algorithms}. This allows us to use a collection of highly uncorrelated algorithms to study different perspectives of a multiplex-community.

\subsection{Unstructured Synthetic Random Multiplex}
The simplest construction of a random multiplex is to generate a set of independent graphs on the same vertex set. However many communities detection algorithms are based on some observations of real world multiplexes and these algorithms do not yield interesting results on such random construction. We denote such random multiplexes as \emph{Unstructured Synthetic Random Multiplex} (USRM), and they are analogous to Erd\H{o}s-R\'{e}nyi Graphs where Communities Detection Algorithms should not find any meaningful communities in them.

For simplicity in the numerical experiments there are only two dimensions in the multiplexes such that we can easily generate all six combinations of Erd\H{o}s-R\'{e}nyi, Watts Strogatz \cite{watts1998collective} and Barab\'{a}si-Albert graphs \cite{albert2002statistical}. The statistical properties of such construction can be found in \cite{loe13}.

\begin{tabular}{lll}
\centering
Name & Graph 1 & Graph 2 \\ \hline
USRM1     & Erd\H{o}s-R\'{e}nyi  &  Erd\H{o}s-R\'{e}nyi       \\
USRM2     & Erd\H{o}s-R\'{e}nyi  &  Watts Strogatz       \\
USRM3     & Erd\H{o}s-R\'{e}nyi  &  Barab\'{a}si-Albert      \\
USRM4     & Watts Strogatz        &  Watts Strogatz       \\
USRM5     & Watts Strogatz        &   Barab\'{a}si-Albert      \\
USRM6     & Barab\'{a}si-Albert  &  Barab\'{a}si-Albert \\
&&
\end{tabular}

In the experiments the size of the vertex set is 128. It is important that the number of edges in both graphs are equal so that neither dominates the interactions of the multiplex. To ensure the Erd\H{o}s-R\'{e}nyi  graph is connected with high probability, the probability that a vertex pair is connected has to be $> \ln{|V|}/|V|$. Therefore the number of edges in the Erd\H{o}s-R\'{e}nyi graph (as well as the other graphs in the multiplex) is ${128 \choose 2}2\ln{128}/128 \approx 616$.

\subsection{Structured Synthetic Random Multiplex}
This construction is similar to the Girvan and Newman Model where independently generated communities are connected in a way such that the ``ground-truth" communities remains. However we saw above that there are different perspective of multiplex communities and we want to encapsulate these ideas into a multiplex benchmark. 

\emph{Structured Synthetic Random Multiplex} (SSRM) is a construction where different multiplex-partitions give high quality multiplex-communities defined by different definitions. However at the same time the partition has to be of ``less-than-ideal" quality from the perspective of other definitions of multiplex-communities.

We begin by generating multiplex-communities, each of which are of high quality with respect to different independent definitions of communities. Next we modify these multiplex-communities such that it remains good in one of the three multiplex-communities definitions \textbf{and} poor quality by the other definitions. The final step is to combine these multiplex-communities into a single multiplex as our SSRM benchmark (Fig. \ref{fig:SSRM}).

Fig. \ref{fig:SSRM} shows $\{[c_1,c_2],[c_3,c_4]\}$ as a partition with 2 communities where clusters 1 and 2 form a community and cluster 3 and 4 form the second community. This partition has high redundancy, low modularity and low CLECC multiplex-communities. There is another partition $\{[c_1,c_3],[c_2,c_4]\}$ where the communities have high modularity, but low in the remaining metrics. Lastly $\{[c_2,c_3],[c_1,c_4]\}$ is the final partition where communities have high CLECC vertex pairs. This synthetic multiplex expresses the multi-perspective nature of a multiplex communities.

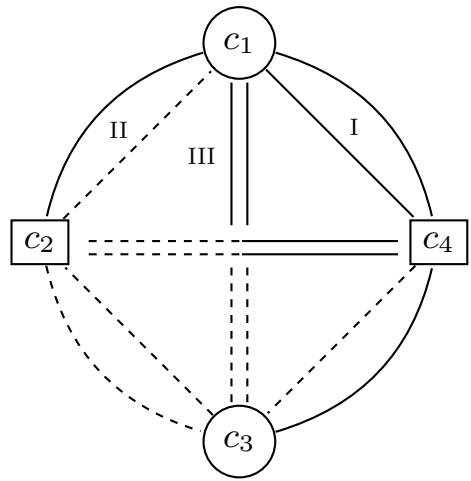
\begin{figure}
  \centering
\begin{tikzpicture}[>=stealth',shorten >=1pt,auto,node distance=2.5cm,thick]

  \tikzstyle{c_node}=[circle, draw,scale=1.5]
  \tikzstyle{r_node}=[rectangle, draw,scale=1.5]

  \node[c_node] (1) {$c_1$};
  \node[r_node] (2) [below left of=1] {$c_2$};
  \node[c_node] (3) [below right of=2] {$c_3$};
  \node[r_node] (4) [below right of=1] {$c_4$};
  \node[draw=none,fill=none] at (-0.5,-1.5) {III};

  \path[thick]
    (1) edge node {I} (4)
    (1) edge [bend left] node {} (4)
    (1) edge [bend right] node {II} (2)
    (3) edge [bend right] node {} (4);

  \path[thick,dashed]
    (2) edge node [right] {} (1)
    (2) edge [bend right] node {} (3)
    (3) edge [right] node {} (2)
    (4) edge node [left] {} (3);

   \draw (-3pt,-15pt) -- (-3pt,-70pt);
   \draw [dashed] (-3pt,-85pt) -- (-3pt,-140pt);
   \draw (3pt,-15pt) -- (3pt,-70pt);
   \draw[dashed] (3pt,-85pt) -- (3pt,-140pt);
        
   \draw (60pt,-80pt) -- (0pt,-80pt);
   \draw [dashed] (0pt,-80pt) -- (-60pt,-80pt);
   \draw (60pt,-75pt) -- (0pt,-75pt);
   \draw[dashed] (0pt,-75pt) -- (-60pt,-75pt); 
\end{tikzpicture}
    \caption{To aid visualization, the edges in this two dimensional SSRM are drawn with solid and dashed lines. Let clusters $c_1$ and $c_3$ be dense subgraphs where there are more solid edges than dashed edges. Similarly $c_2$ and $c_4$ are dense subgraphs with more dashed edges than solid edges. \textbf{I:} The solid edges between $c_1$ \& $c_4$ implies that there are only solid edges between them. This applies the same to the dashed edges between $c_2$ \& $c_3$. \textbf{II:} All the edges between $c_1$ and $c_2$ overlaps. \textbf{III:} None of the edges between $c_1$ \& $c_3$ (or $c_2$ \& $c_4$) overlaps. We denote $\{[c_1,c_2],[c_3,c_4]\}$ as a partition with 2 communities where clusters 1 and 2 form a community and cluster 3 and 4 form the second community. This partition has high redundancy, low modularity and low CLECC multiplex-communities.\label{fig:SSRM}}
\end{figure}

\subsubsection{High Modularity, Low Redundancy \& Low CLECC Multiplex-Communities}
To create a high modularity multiplex-community, we begin with cluster 1 and cluster 3 as a clique-like structure in both dimensions. For these clusters to be low in redundancy, we have to remove some edges in both clusters such that there are very few overlapping edges in the clusters while maintaining high modularity. This can be done by making the first dimension graph in both clusters to be the complement of the graph in the second dimension.

The next step is to add edges between cluster 1 and 3 such that the resultant cluster is a single connected component. We denote $[c_1,c_3]$ as a component connected by cluster 1 and cluster 3. To maintain a low redundancy, the new edges cannot overlap.

Finally we need to tweak the clusters such that the CLECC score is low between a significant number of the vertex pairs in the combined component of cluster 1 and 3. Specifically we want the vertex pairs connected by the new edges in the previous step to have low CLECC scores. This is possible if cluster 1 has more edges in the first dimension whereas cluster 3 has more edges in the second dimensions. In doing so the neighbors of the vertex in cluster 1 will be significantly different from the neighbors of vertex in cluster 3, thus a low CLECC score.

The same construction applies to cluster 2 and cluster 4, where cluster 2 is similar to cluster 3 and cluster 4 is similar to cluster 1. 

\subsubsection{High CLECC, Low Modularity \& Low Redundancy Multiplex-Communities}
Since all 4 clusters do not have overlapping edges, the redundancy remains low for any partition on the multiplex. Therefore the first step is to make cluster 1 and 4 to be high in CLECC score, and apply the same construction for cluster 2 and 3.

Since cluster 1 and 4 are similarly based on the construction from the previous subsection, the neighbors of any given vertex in each cluster will be similar too. Therefore by adding new edges between cluster 1 and 4 will not affect the CLECC score. However these new edges should \emph{only} be drawn in the first dimension, since it is the dominant dimension in $[c_1,c_4]$. This simultaneously reduces the modularity of $[c_1,c_4]$ in the second dimension, since the clusters are not connected and the graph in the second dimension is sparse. This gives a low modularity for multiplex communities while maintaining the high CLECC score.

The construction is similar for cluster 2 and 3, expect that only edges in the second dimension connects both clusters together.

\subsubsection{High Redundancy, Low Modularity \& CLECC Multiplex-Communities}
Given $[c_1,c_3]$ have low CLECC score, the same score should apply for $[c_1,c_2]$ since cluster 2 and cluster 3 are similar. The main goal is to connect cluster 1 and cluster 2 such that the redundancy of $[c_1,c_2]$ is high. Redundancy is measured by Eq. \ref{def:redundancy}, where it counts the number of edges that overlaps. Since there is no overlapping edges at this point of the construction, the simplest way to increase the redundancy is to add new overlapping edges between clusters to form the components $[c_1,c_2]$ and $[c_3,c_4]$.

Although $[c_1,c_2]$ and $[c_3,c_4]$ have relatively high redundancy as compare to other partition in the multiplex, it can still have lower redundancy than a random community in USRM1. To nudge the redundancy higher, it is necessary to add new edges such that there are overlaps in the four clusters. However this might increase the modularity of $[c_1,c_2]$ which we want to avoid. Therefore this final step has to be done incrementally.

\subsubsection{Evaluation of the different ground truth partitions}
There are 3 different ``ground-truth" partitions $\{[c_1,c_2],[c_3,c_4]\}$, $\{[c_2,c_3],[c_1,c_4]\}$ and $\{[c_1,c_3],[c_2,c_4]\}$, where they each represents a different ``ideal-partition". However simultaneously they are ``less-than-ideal" from the perspective of the other metrics. Although some of these metrics are correlated (from our experiments) in general, these ground-truth partitions of SSRM show that these metrics are each able to capture essential aspects of the community structure.

\begin{table}[h]
\begin{tabular}{lrrr} 
   & Redundancy & *CLECC & **Modularity \\ \hline
$[c_1,c_2]$ & \textbf{0.0492}     & 0.1142    & -0.0287     \\
$[c_3,c_4]$ & \textbf{0.0537}     & 0.1087    & -0.0332     \\ \hline
$[c_2,c_3]$ & 0          & \textbf{0.1541}    & 0.007       \\
$[c_1,c_4]$ & 0          & \textbf{0.1642}    & 0.012       \\ \hline
$[c_1,c_3]$ & 0          & 0.1113    & \textbf{0.0317}      \\
$[c_2,c_4]$ & 0          & 0.1083    & \textbf{0.0245}      \\ \hline
Random & 0.0217     & 0.1056    & 0.0033     
\end{tabular}
\caption{The rows (except the last row) are paired up such that they implies a common partition. For example the first two rows are communities of the partition $\{[c_1,c_2],[c_3,c_4]\}$. The redundancy and CLECC score of community $[c_1,c_2]$ are 0.0492 and 0.1142 respectively. \textbf{*CLECC:} The CLECC score is the average CLECC score between all vertex pairs in the community. \textbf{**Modularity:} The two values in the partition refers to the modularity of the two graphs in the multiplex. E.g. partition $\{[c_1,c_2],[c_3,c_4]\}$ has modularity -0.0287 and -0.0332 for the first and second dimensions of SSRM respectively. A partition is ``less-than-ideal" if its measurement is closer to a random partition (last row) than the maximum (values in bold).
\label{table:SSRM}}
\end{table}

For example Table \ref{table:SSRM} shows that the partition $\{[c_1,c_2],[c_3,c_4]\}$ has communities with the maximum redundancy. However it has multi-modularity and CLECC scores similar to a random partition, suggesting that it is ``less-than-ideal" relative to other metrics.

\subsection{Real World Multiplex}
The issue with synthetic multiplex is that it does not reflect the real-world systems correctly. The relationships between vertices are artificially imposed such that we can distinguish the different communities detection algorithms. However real-world communities do not behave in such a systematic manner and might approach similar structure and relationships even if they are created by different dynamics. Therefore we compare the multiplex-communities detection algorithms with real-world multiplexes from open dataset.

\subsubsection{Youtube Social Network}
Youtube is a video sharing website that allows interactions between the video creators and their viewers. Tang et al. collected 15,088 active users to form a multiplex with 5 relationships where two users are connected if \cite{tang2009}:
\begin{enumerate}
  \item they are in the contact list of each other;
  \item their contact list overlaps;
  \item they subscribe to the same user's channel;
  \item they have subscription from a common user;
  \item they share common favorite videos.
\end{enumerate}

\subsubsection{Transportation Network}
Generalized graphs in transportation networks are known as multimodal network, where bus stops, train stations and terminals are indistinguishable locations (vertices) to transit to a different transport. Cardillo et al. constructed an air traffic multiplex from the data of European Air Transportation (EAT) Network with 450 airports as vertices \cite{cardillo2013emergence}. An edge is drawn between two vertices if there is a direct flight between them. Each of the 37 distinct airlines in the EAT Network forms a relationship between airports.

Given that the base of airlines are often in older and major airports, hence there is the ``rich-gets-richer" phenomenon where new airports tends to have direct flight to major airports. Therefore Cardillo et al. verified the phenomenon by ensuring that the degree distribution follows a power-law distribution, implying a scale-free system.

\section{Comparing Partitions}
Normalized Mutual Information (NMI) \cite{manning2008introduction} is a popular similarity metric for the network partition, where a real-value score between $[0,1]$ with score 1 implying identical partition. However we are not able to simply apply the NMI metric to all pairwise comparisons of all the algorithms.

Unlike the rest of the algorithms, \emph{Frequent Closed Itemsets Mining} and \emph{Tensor Decomposition} yield overlapping communities, which NMI was not designed to measure. Furthermore these algorithms also do not ensure that all vertices belong in at least one community, i.e. there might exists ``isolated" vertices with no community membership. Therefore even variants of NMI \cite{lancichinetti2009detecting} for overlapping communities are not applicable in our study.

Adjusted Rand Index is an alternative metric to NMI in classification research \cite{hubert1985comparing}. Its extension to measure overlapping communities is known as Omega Index \cite{collins1988omega} and is able to measure partitions with isolated vertices. Unlike overlapping variants of NMI, the comparison of two non-overlapping partitions with Omega Index reduces to Adjusted Rand Index.

\subsection{Normalized Mutual Information}
Mutual Information $I(\mathcal{A};\mathcal{B})$ measures the information of the communities-membership of all vertex-pairs  in $A$ given the communities-membership in $B$, vice versa. Roughly speaking given $\mathcal{B}$, how well can we guess that a vertex pair is in the same community in $\mathcal{A}$? Formally this is defined as $I(\mathcal{A};\mathcal{B}) = H(\mathcal{A}) - H(\mathcal{A}|\mathcal{B})$ where $H(\mathcal{A})$ is the Shannon entropy:
\begin{equation}\label{NMI_den}
H(\mathcal{A}) = -\sum_k P(a_k) \log P(a_k).
\end{equation}
where $P(a_k)$ is the probability of a random vertex is in the $k^{th}$ community of partition $\mathcal{A}$, i.e. the ratio of the size of the $k^{th}$ community to the total number of vertices. Similarly $P(b_k)$ denotes the case for partition $\mathcal{B}$. The Mutual Information can also be expressed as:

\begin{equation}\label{NMI_num}
I(\mathcal{A} ; \mathcal{B}) = \sum_j \sum_k P(a_k \cap b_j) \log
\frac{P(a_k \cap b_j)}{P(a_k)P(b_j)}, 
\end{equation}
where $P(a_k \cap b_j)$ is the probability that a random vertex is both in $k^{th}$ and $j^{th}$ communities. Basically the larger the intersection of the $k^{th}$ and $j^{th}$ communities of $\mathcal{A}$ and $\mathcal{B}$ respectively is, the higher the Mutual Information. Finally to normalize the Mutual Information score:
\begin{equation}\label{NMI}
\mbox{NMI}(\mathcal{A} , \mathcal{B})
=
\frac{
I(\mathcal{A} ; \mathcal{B})
}
{
[H(\mathcal{A})+ H(\mathcal{B})]/2.
}
\end{equation}

\subsection{Omega Index}
The unadjusted Omega Index averages the number of vertex pairs that are in the \emph{same number} of communities. Such vertex pairs are known to be in \emph{agreement}. Consider the case with two partitions $\mathcal{A}$ and $\mathcal{B}$, and the number of communities in them are $|\mathcal{A}|$ and $|\mathcal{B}|$ respectively. The function $t_j(\mathcal{A})$ returns the set of vertex pairs that appears exactly in $j \geq 0$ overlapping communities in $\mathcal{A}$. Thus the unadjusted Omega Index:
\begin{equation}
\omega_u(\mathcal{A};\mathcal{B}) = \frac{1}{{n \choose 2}} \sum_{j=0}^{max(|\mathcal{A}|,|\mathcal{B}|)} |t_j(\mathcal{A}) \cap t_j(\mathcal{B})|.
\end{equation}

To account for vertex pairs that are allocated into the same communities by chance, we have to subtract it from expected omega index of a null model:
\begin{equation}
\omega_e(\mathcal{A};\mathcal{B}) = \frac{1}{{n \choose 2}^2} \sum_{j=0}^{max(|\mathcal{A}|,|\mathcal{B}|)} |t_j(\mathcal{A})| \cdot |t_j(\mathcal{B})|.
\end{equation}

Finally the Omega Index is given by:
\begin{equation}
\omega(\mathcal{A};\mathcal{B}) = \frac{\omega_u(\mathcal{A};\mathcal{B}) - \omega_e(\mathcal{A};\mathcal{B})}{1 - \omega_e(\mathcal{A};\mathcal{B})}.
\end{equation}

It is possible for the Omega Index to be negative, where there are less agreement than pure stochastic coincidence would expect. It is regarded as uninteresting as it does not suggests anything more than the fact that the partitions are not similar. Identical partitions have Omega Index of 1.

\subsection{Notations For Empirical Results}
To simplify the plots exhibiting the results from our experiments, we will use some shorthands to denote the algorithms and partitions. For communities detection algorithms, we use $\mathcal{A}$ and $\mathcal{P}$ to denote ``Algorithm" and ``SSRM Partition" respectively (Table \ref{table:notation}).

\begin{table}[h]
\begin{tabular}{ll} 
Notation & Description \\ \hline
$\mathcal{A}_1$ & Projection-Binary \\
$\mathcal{A}_2$ & Projection-Average \\
$\mathcal{A}_3$ & Projection-Neighbors \\
$\mathcal{A}_4$ & Cluster-based Similarity Partition Algorithm \\
$\mathcal{A}_5$ & Generalized Canonical Correlations \\
$\mathcal{A}_6$ & CLECC Bridge Detection\\
$\mathcal{A}_7$ & Frequent Closed Itemsets Mining \\
$\mathcal{A}_8$ & PARAFAC, Tensor Decomposition\\
$\mathcal{P}_1$ & $\{[c_1,c_2],[c_3,c_4]\}$ of SSRM \\
$\mathcal{P}_2$ & $\{[c_2,c_3],[c_1,c_4]\}$ of SSRM \\
$\mathcal{P}_3$ & $\{[c_1,c_3],[c_2,c_4]\}$ of SSRM 
\end{tabular}
\caption{Shorthands for the different algorithms and ``ground-truth" communities.
\label{table:notation}}
\end{table}

\section{Empirical Comparison of Multiplex-Communities Detection Algorithms} \label{sec:comparison}
\subsection{Unstructured Synthetic Random Multiplex}
Figure \ref{fig:USRMplot_Omega} shows the Omega Index of all pairwise multiplex communities detection algorithms for the USRM benchmarks. The last 13 boxplots on the right are the pairwise comparisons with overlapping-communities algorithms, i.e. $\mathcal{A}_7$ and $\mathcal{A}_8$. 

The first observation is that $\mathcal{A}_8$ (PARAFAC) is not similar to any of the algorithms. One of the reasons is that it is hard to choose the right parameters for PARAFAC, i.e the $k$ approximation for Eq. \ref{Eq:PARAFAC} and the predefined threshold for the top few elements of the rank-one tensors. There is no systematic method besides manual experimentation for a given graph \cite{dunlavy,Judith,Andri,papalexakis}.

However overlapping-communities algorithm $\mathcal{A}_7$ (Frequent Closed Itemsets Mining) is similar to a class of non-overlapping algorithms, i.e. $\mathcal{A}_1$ to $\mathcal{A}_4$. Specifically Frequent Closed Itemsets Mining is similar to the class of Projection algorithms $\mathcal{A}_1$ to $\mathcal{A}_3$. In fact the Omega Index for all pairwise comparisons of the algorithm family $\mathcal{F} = \{\mathcal{A}_1, \mathcal{A}_2,\mathcal{A}_3,\mathcal{A}_4,\mathcal{A}_7\}$ is generally greater than the other algorithm pairs. This is supported by their NMI scores in Fig. \ref{fig:USRMplot}.

However the family of algorithms $\mathcal{F}$ has low pairwise Omega Index and NMI scores for USRM 1, 3 and 6. There is no fundamental reasons to why the class of projection algorithms ($\mathcal{A}_1$ to $\mathcal{A}_3$) should produce non similar communities, therefore we can deduce that USRM 1, 3 and 6 are not good multiplexes for benchmarks.

The reason is that USRM1 is the combination of two  Erd\H{o}s-R\'{e}nyi graphs, thus there is no community structure. Whereas USRM 3 and 6 are the combinations of  Barab\'{a}si-Albert graph with Erd\H{o}s-R\'{e}nyi  and  Barab\'{a}si-Albert respectively. Although  Barab\'{a}si-Albert has structural properties, it has low Clustering Coefficient (similar to Erd\H{o}s-R\'{e}nyi), which measures the tendency that vertices cluster together. Therefore the vertices in USRM 3 and 6 do not form communities.

To increase the clustering coefficient of a multiplex, we can introduce a Watts Strogatz graph to the system. We can approximate the clustering coefficient of the projection of such multiplexes \cite{loe13} and deduce the tendency that communities exists. This allows us to observe interesting relationships between the algorithms from benchmark multiplexes like USRM 2, 4 and 5. 

For example USRM 5 in Fig. \ref{fig:USRMplot} shows that the similarity range of $\mathcal{F}$ with $\mathcal{A}_6$ (CLECC Bridge Detection) is wide. This is more apparent when we study their relationship with SSRM benchmark.

\begin{figure}
  \centering
  \includegraphics[width=80mm]{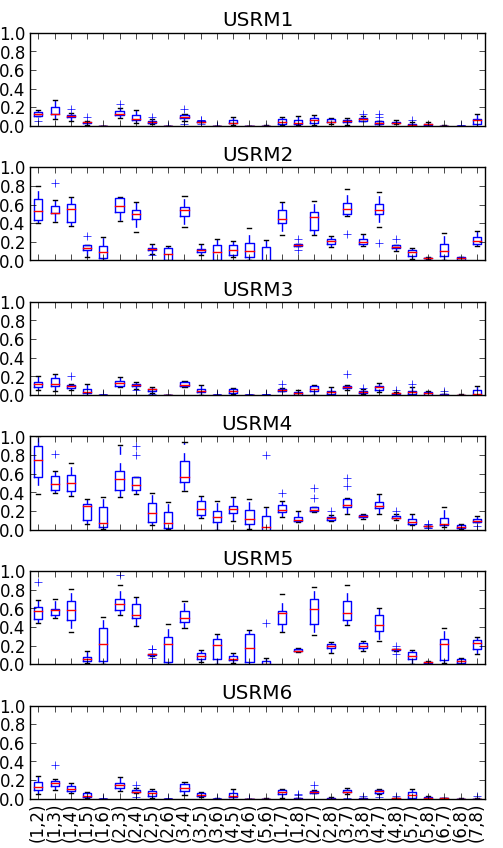}
    \caption{The Omega Index of all pairwise multiplex-communities detection algorithms for the different benchmark USRM. The tuple $(i,j)$ on the x-axis refers to the pairwise comparison of $\mathcal{A}_i$ and $\mathcal{A}_j$. The tuples are arranged such that the comparisons with overlapping-communities algorithms (13 tuples) are placed on the right. \label{fig:USRMplot_Omega}}
\end{figure}

\begin{figure}
  \centering
  \includegraphics[width=85mm,height=135mm]{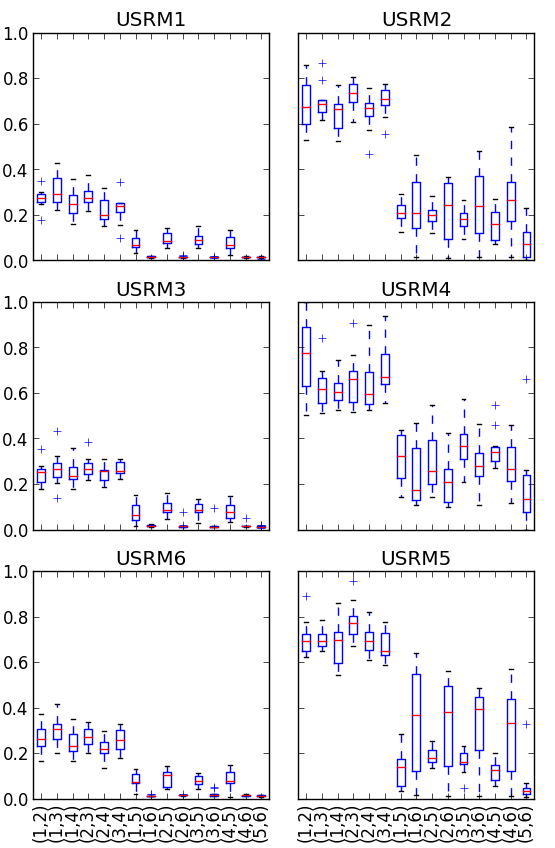}
    \caption{The NMI scores of all pairwise non-overlapping multiplex-communities detection algorithms for USRM benchmarks. The tuples are arranged such that pairwise comparisons of $\{\mathcal{A}_1, \mathcal{A}_2,\mathcal{A}_3,\mathcal{A}_4\}$ are grouped to the left of the boxplots. The ``interesting" figures USRM 2, 4 and 5 are placed to the right for comparison.\label{fig:USRMplot}}
\end{figure}

\subsection{Structured Synthetic Random Multiplex}
In the previous section, USRM benchmark suggests that $\{\mathcal{A}_1, \mathcal{A}_2,\mathcal{A}_3,\mathcal{A}_4\}$ yields similar communities. However Fig. \ref{fig:SSRMplot} shows that the SSRM communities from $\mathcal{A}_4$ (Cluster-based Similarity Partition Algorithm) are distinct from the class of projection algorithms ($\mathcal{A}_1$ to $\mathcal{A}_3$). Therefore we can conclude that Cluster-based Similarity Partition Algorithm provides an alternative perspective for multiplex-communities, and it is not reducible to the class of projection algorithms.

Recall from the last observation in the previous section that the similarity range of $\mathcal{A}_6$ (CLECC Bridge Detection) with projection algorithms is wide. Fig. \ref{fig:SSRMplot} shows that there are cases where they are very different (close to NMI = 0), and there are cases where the NMI $>0.8$. In this way our experiments highlights the weakness of the CLECC Bridge Detection algorithm.

CLECC occasionally yields separate components prematurely, where one of the components is a small cluster of vertices or even a single vertex as a community. Therefore it appears that CLECC yields significantly different communities since the rest of the algorithms tend to produce balanced-size communities. Hence if we exclude such cases for CLECC, $\mathcal{A}_6$ is similar to projection algorithms for SSRM.

Finally we will compare the algorithms with the ``ground-truth" partitions $\mathcal{P}_1$, $\mathcal{P}_2$ and $\mathcal{P}_3$. Table \ref{table:SSRM_GT} shows that none of the algorithms were able to capture $\mathcal{P}_1$, which is the partition with high redundancy. $\mathcal{A}_3$ (Projection-Neighbors) was proposed to extract high redundancy communities \cite{Berlingerio11_2}.

In contrast, $\mathcal{A}_6$ (CLECC Bridge Detection) was able to find the high CLECC partition $\mathcal{P}_2$. However it does not have any advantage over any of the projection algorithms ($\mathcal{A}_1$ to $\mathcal{A}_3$). This further supports that $\mathcal{A}_6$ is similar to the class of projection algorithms.

The results by Tang et. al shows that $\mathcal{A}_5$ (Generalized Canonical Correlations) tends to be better than $\mathcal{A}_4$ (Cluster-based Similarity Partition Algorithm)  to capture high-modularity communities like $\mathcal{P}_3$ \cite{tang2009}. This was also observed in this experiment.

\begin{figure}
  \centering
  \includegraphics[width=80mm]{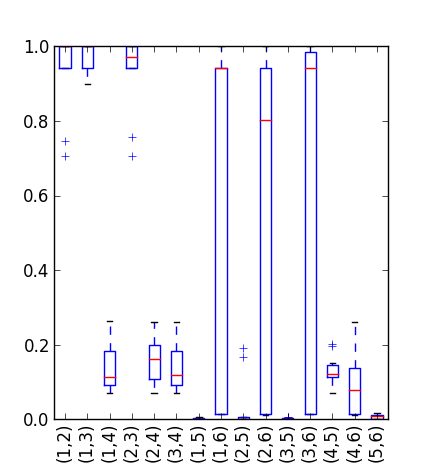}
    \caption{The NMI scores of all pairwise non-overlapping multiplex-communities detection algorithms for SSRM benchmarks. The tuples are arranged in the same way as Fig. \ref{fig:USRMplot}. \label{fig:SSRMplot}}
\end{figure}

\begin{table}[h]
\setlength{\tabcolsep}{1em}
\begin{tabular}{lrrr}
& $\mathcal{P}_1$    & $\mathcal{P}_2$    & $\mathcal{P}_3$    \\ \hline
$\mathcal{A}_1$ & 0     & \textbf{0.983} & 0     \\
$\mathcal{A}_2$ & 0.002 & \textbf{0.94}  & 0.017 \\ 
$\mathcal{A}_3$ & 0     & \textbf{0.978} & 0     \\ 
$\mathcal{A}_4$ & \textbf{0.019} & 0.14  & 0.083 \\ 
$\mathcal{A}_5$ & 0.004 & 0.002 & \textbf{0.158} \\ 
$\mathcal{A}_6$ & 0.006 & \textbf{0.964} &  0.006 \\ 
\end{tabular}
\caption{The NMI scores between the algorithms and the different ground-truth partitions. The entries in bold represents the algorithms that are ``close" to the ground-truth partition.
\label{table:SSRM_GT}}
\end{table}

\subsection{Real World Multiplex}
The results for the European Air Transportation Network is similar to the Youtube Social Network, hence Fig. \ref{fig:AirlinePlot} is sufficient for this discussion. The general observation is similar to USRM 2, 4 and 5 in Fig. \ref{fig:USRMplot_Omega}, where the set of algorithms $\{\mathcal{A}_1,\mathcal{A}_2,\mathcal{A}_3,\mathcal{A}_4,\mathcal{A}_7\}$ are relatively similar with pairwise NMI score of $\approx 0.55$.

In addition, Fig. \ref{fig:AirlinePlot} highlights that overlapping-communities detection algorithm $\mathcal{A}_8$ (PARAFAC) is relatively more similar to $\mathcal{A}_5$ (Generalized Canonical Correlations) than the other non-overlapping communities detection algorithms. This observation was less apparent in Fig. \ref{fig:USRMplot_Omega}.

Unfortunately $\mathcal{A}_6$ (CLECC Bridge Detection) tends to halt prematurely despite different parameter choices. Hence we did not managed to get any insight for CLECC Bridge Detection in this experiment.

\begin{figure}
  \centering
  \includegraphics[width=80mm]{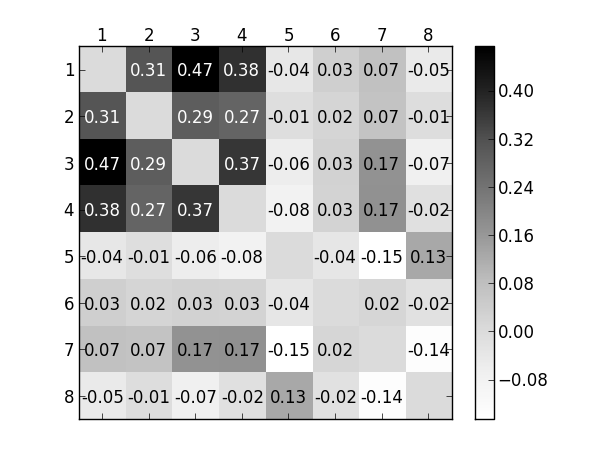}
    \caption{The Omega Index heatmap of all pairwise algorithms for the European Air Transportation Network. This result is similar to the Youtube Social Network. \label{fig:AirlinePlot}}
\end{figure}

\subsection{Summary}
The parameters in algorithms $\mathcal{A}_6$, $\mathcal{A}_7$ and $\mathcal{A}_8$ require manual fine-tuning to yield meaningful communities for comparisons. Hence it is not practical to exhaustively test for all configurations for these algorithms. However the out come of the analysis doesn't change in any essential way when for different choice of parameters.

From USRM benchmarks and real world multiplexes, algorithms in the set $\{\mathcal{A}_1,\mathcal{A}_2,\mathcal{A}_3,\mathcal{A}_4,\mathcal{A}_7\}$ tends to generate relatively similar partitions. However SSRM benchmark demonstrate that $\mathcal{A}_4$ is able to capture high redundancy communities and performed differently from the class of projection algorithms.


Our experiments with USRM and SSRM benchmarks support that $\mathcal{A}_6$ (CLECC Bridge Detection) is similar to the class of projection algorithms $\{\mathcal{A}_1,\mathcal{A}_2,\mathcal{A}_3\}$ when the algorithm does not infer small clusters of vertices as communities. Therefore CLECC Bridge Detection is not very stable without careful analysis.

Finally we have to emphasize that the absolute Omega Index and NMI index are generally low ($<0.5$) for most pairwise comparisons. This implies that the algorithms technically perceive less-than-ideal multiplex-communities differently. However our results shows that there are some vertex clusters within communities that are more likely to be captured by certain algorithms than the others.

\section{Discussions And Conclusion}
The literature review and SSRM benchmark highlight the additional complexity to define a multiplex-community. It is a balance between the structural topology of the communities and their relationships. SSRM is a useful benchmark to bring out the fundamental differences between these algorithms. Therefore further research is done to improve SSRM such that the different definitions of communities are more apparent.

The main contribution of this paper is to consolidate and compare all the existing multiplex-communities detection algorithms. The long list of synonymous names for multiplex causes many researchers to be unaware of related efforts for multiplex, and hence has a tendency to make this research some what diffuse and fragmented \cite{Kivela2013}. Therefore we hope that this paper brings awareness for further developments on multiplex-communities detection algorithm such that researchers can build on to existing ideas.

\bibliography{database}

\end{document}